\documentclass[11pt]{article}

\usepackage[utf8]{inputenc}
\usepackage[T1]{fontenc}
\usepackage{lmodern}
\usepackage{microtype}

\usepackage{amsmath}
\usepackage{amssymb}
\usepackage{amsthm}
\usepackage{mathtools}

\usepackage{geometry}
\usepackage{enumitem}

\usepackage[
  colorlinks=true,
  linkcolor=blue,
  citecolor=blue,
  urlcolor=blue
]{hyperref}

\numberwithin{equation}{section}

\newtheorem{theorem}{Theorem}[section]
\newtheorem{proposition}[theorem]{Proposition}
\newtheorem{lemma}[theorem]{Lemma}

\theoremstyle{definition}

\theoremstyle{remark}
\newtheorem{remark}[theorem]{Remark}

\title{A Complete Classification of Low-Order Conservation Laws
for a Generalized Fifth-Order KP Family}

\author{
Nitin Serwa\\Abu Dhabi University
}

\date{}

\begin{document}

\maketitle

\begin{abstract}
A complete classification of local conservation laws with multipliers
of differential order at most two is obtained for a generalized
fifth-order Kadomtsev--Petviashvili family. The classification is
carried out by the direct multiplier method and concerns nonlinear
members of the family, with the conservation laws expressed locally
in the original dependent variable. It is first shown, uniformly in
the parameters, that every multiplier of differential order at most
two reduces to first order. The resulting determining equations yield
one generic case and two exceptional nonlinear cases. In the generic
case, the multipliers involve four arbitrary functions of time. One
exceptional case admits an additional multiplier depending on the
first longitudinal derivative and involves five arbitrary functions
of time, while the other has an enlarged zeroth-order multiplier
family involving eight arbitrary functions of time. Representative conserved densities and spatial fluxes are derived for
all three cases. The generic nonzero densities represent mass and
transverse moment-type quantities. The first exceptional case admits
the longitudinal gradient-energy density $\tfrac12 u_x^2$, whereas
the second admits higher longitudinal moment densities. No multiplier
proportional to $u$ occurs within the classified low-order local
family, and hence no $L^2$-type density arises within this
classification. The corresponding conserved integrals are obtained
under appropriate boundary conditions or sufficient weighted spatial
decay.
\end{abstract}

\medskip


\medskip


\section{Introduction}
\label{sec:introduction}

The Kadomtsev--Petviashvili (KP) equation was introduced to describe
the stability of one-dimensional solitary waves under weak transverse
perturbations \cite{KadomtsevPetviashvili1970}. It has since become a
basic model for weakly nonlinear dispersive waves in two spatial
dimensions and has been studied extensively from the viewpoints of
integrable systems, inverse scattering, solitary waves, and the
Cauchy problem
\cite{AblowitzClarkson1991,Konopelchenko1993,KleinSaut2021}.
Higher-order longitudinal dispersion leads to generalized KP
equations whose analytical properties differ substantially from those
of the standard third-order KP equation. The Cauchy problem for higher-order and fifth-order KP equations has
been considered in
\cite{SautTzvetkov1999,SautTzvetkov2000}. More recently,
unconditional uniqueness for fifth-order KP-type equations has been
studied in \cite{Patterson2025}, while transverse stability questions
for fifth-order solitary waves have been considered in
\cite{Karpman1993}.

In one spatial dimension, fifth-order dispersive equations include
the Lax, Sawada--Kotera, Kaup--Kupershmidt, and Kawahara equations
\cite{Lax1968,SawadaKotera1974,Kaup1980,Kupershmidt1985,Kawahara1972}.
The first three equations are integrable fifth-order equations with
different nonlinear coefficient combinations, whereas the Kawahara
equation models the effect of competing third- and fifth-order
dispersion. Two-dimensional extensions of fifth-order equations have
been studied by inverse spectral and dressing methods
\cite{Konopelchenko1993,DubrovskyTopovskyBasalaev2010}, and
higher-dimensional equations involving extended shallow-water
dispersion have arisen in the modelling of resonant wave phenomena
\cite{HorikisFrantzeskakisMarchantSmyth2021}. A fifth-order
Kawahara--KP equation has also been investigated from the viewpoints
of symmetries, conservation laws, and line solitary waves
\cite{MarquezLozanoGandariasAnco2023}. These works provide the
one- and two-dimensional context for the family considered here; the
one-dimensional reductions are used only for context and not as an
integrability criterion for the two-dimensional equation.

We consider the generalized fifth-order KP family
\begin{equation}
\bigl(
u_t
+c\,u_{xxx}
+d\,u_{xxxxx}
+e\,u\,u_{xxx}
+f\,u_xu_{xx}
+g\,u^2u_x
\bigr)_x
+\sigma u_{yy}
=0,
\label{eq:intro-equation}
\end{equation}
where $d\neq0$ and $\sigma\neq0$. For solutions that are independent
of $y$, equation~\eqref{eq:intro-equation} gives
\begin{equation*}
D_x\left(
u_t
+c\,u_{xxx}
+d\,u_{xxxxx}
+e\,u\,u_{xxx}
+f\,u_xu_{xx}
+g\,u^2u_x
\right)=0.
\end{equation*}
After one integration with respect to $x$, this becomes
\begin{equation}
u_t
+c\,u_{xxx}
+d\,u_{xxxxx}
+e\,u\,u_{xxx}
+f\,u_xu_{xx}
+g\,u^2u_x
=h(t),
\label{eq:one-dimensional-reduction}
\end{equation}
where $h(t)$ is an integration function. Under spatial decay, $h(t)=0$. More generally, we fix the integration
function to be zero when identifying the standard one-dimensional
reductions.

With $c=0$ and $d=1$, equation~\eqref{eq:one-dimensional-reduction}
contains the Lax fifth-order KdV equation for
\begin{equation*}
(e,f,g)=(10,20,30),
\end{equation*}
the Sawada--Kotera equation for
\begin{equation*}
(e,f,g)=(5,5,5),
\end{equation*}
and the Kaup--Kupershmidt equation for
\begin{equation*}
(e,f,g)=(10,25,20),
\end{equation*}
in the sign convention used here
\cite{Lax1968,SawadaKotera1974,Kaup1980,Kupershmidt1985}.
These reductions motivate the nonlinear terms in
\eqref{eq:intro-equation}. They are included only as context: the
classification below is uniform in the parameters and does not use
the integrability of a one-dimensional reduction as a criterion for
the two-dimensional equation. The quadratic convection term $u u_x$ is not included:
the family retains precisely the three nonlinear terms
$u u_{xxx}$, $u_x u_{xx}$, and $u^2 u_x$ occurring in the Lax,
Sawada--Kotera, and Kaup--Kupershmidt reductions, rather than to
represent the full polynomial class closed under translations
$u \mapsto u+\kappa$.

Conservation laws are fundamental in the study of nonlinear partial
differential equations. They give local balance relations and, under
suitable boundary or decay conditions, conserved integrals. They are
also useful in the analysis of solutions and in assessing numerical
approximations. A local conservation law is represented by a
divergence expression that vanishes on all solutions, with conserved
density and spatial fluxes depending locally on the independent
variables, the dependent variable, and finitely many derivatives.
Two conservation laws are regarded as equivalent when they differ by
a locally trivial conservation law.

For equations that do not necessarily possess a Lagrangian
formulation, local conservation laws can be found directly through
their multipliers. The direct multiplier method and the general
multiplier--conservation-law correspondence were developed by Anco
and Bluman
\cite{AncoBluman2002a,AncoBluman2002b}; see also
\cite{BlumanCheviakovAnco2010,Olver1993}. In this approach, a
multiplier is determined by requiring that its product with the
differential equation be a total divergence. The method gives
conservation laws up to local equivalence and does not require the
equation to be variational. Related approaches based on adjoint
symmetries and nonlinear self-adjointness are discussed in
\cite{Ibragimov2011,IbragimovAvdonina2013,Ma2018}.

Classifications of conservation laws have been obtained for various
generalized KdV, KP, Boussinesq, Kawahara, and related
higher-dimensional dispersive equations. In particular,
conservation laws, symmetries, and line soliton solutions have been
studied for generalized KP and Boussinesq equations with power
nonlinearities \cite{AncoGandariasRecio2018}. Conservation laws for a
Kawahara--KP equation were derived in
\cite{MarquezLozanoGandariasAnco2023}, while conservation laws and
exact solutions for another higher-dimensional equation with
higher-order dispersion were considered in
\cite{AliSeadawyHusnine2019}. The present family differs from these
equations through its simultaneous inclusion of the three nonlinear
terms
\[
u\,u_{xxx},
\qquad
u_xu_{xx},
\qquad
u^2u_x,
\]
with independent coefficients. Consequently, its conservation-law
classification requires the nonlinear coefficient space to be
treated without fixing a particular integrable or physically
distinguished reduction.

The purpose of this paper is to give a complete classification of
the low-order local conservation laws of the nonlinear members of
\eqref{eq:intro-equation}. Here, low order means that the associated
multipliers are local in the original dependent variable $u$ and have
total differential order at most two. Conservation laws that become
local only after the introduction of a potential are not included.

The classification has three main features. First, every multiplier
of differential order at most two is shown, uniformly in the
parameters, to be independent of all second-order jet variables and
hence to reduce to first order. Second, the unrestricted first-order
determining system is solved completely. For generic nonlinear
coefficients, the multipliers form a family involving four arbitrary
functions of time. Two exceptional nonlinear cases occur. When
\[
g=0,
\qquad
e=2f,
\qquad
f\neq0,
\]
an additional multiplier depending on the longitudinal derivative
$u_x$ is admitted, and the resulting family involves five arbitrary
functions of time. When
\[
g=0,
\qquad
f=3e,
\qquad
e\neq0,
\]
the zeroth-order multiplier can be cubic in $x$, and the family
involves eight arbitrary functions of time. The two exceptional cases
are explained by structural identities in the corresponding
one-dimensional nonlinear terms. Finally, representative conserved
densities and spatial fluxes are derived for all three cases, together
with the associated conserved integrals under suitable boundary,
periodicity, or spatial decay conditions.

The paper is organized as follows.
Section~\ref{sec:equation} gives the divergence form of the equation
and its equivalence transformations.
Section~\ref{sec:multiplier-framework} recalls the multiplier
framework and fixes the scope of the classification.
Section~\ref{sec:order-reduction} proves the reduction from
second-order to first-order multipliers.
Section~\ref{sec:classification} gives the complete first-order
classification, including the generic and two exceptional cases.
Section~\ref{sec:conserved-currents} presents the corresponding
conserved currents, interprets the conserved densities, and discusses
the associated conserved integrals, while
Section~\ref{sec:conclusion} gives concluding remarks.
\section{The equation and equivalence transformations}
\label{sec:equation}

\subsection{Divergence form}

We consider the generalized fifth-order KP family
\eqref{eq:intro-equation}. Since $d\neq0$ and $\sigma\neq0$, scalings
of $x$, $y$, and $t$ can be used to normalize the coefficient of
$u_{xxxxx}$ inside the $x$-divergence to one and the magnitude of the
coefficient of $u_{yy}$ to one. More precisely, the sign of the
normalized transverse coefficient is
\[
\operatorname{sgn}\left(\frac{\sigma}{d}\right).
\]
Reusing $\sigma$ to denote this normalized sign, and dropping the
tildes on the scaled variables, we work throughout with the
normalized family
\begin{equation}
\bigl(u_t+c\,u_{xxx}+u_{xxxxx}
+e\,u\,u_{xxx}+f\,u_xu_{xx}+g\,u^2u_x\bigr)_x
+\sigma u_{yy}=0,
\qquad
\sigma=\pm1.
\label{eq:normalized-equation}
\end{equation}
Its expanded form is
\begin{equation}
\Delta = u_{tx} + c\,u_{xxxx} + u_{xxxxxx}
+ (e+f)\,u_x u_{xxx} + e\,u\,u_{xxxx} + f\,u_{xx}^2
+ 2g\,u\,u_x^2 + g\,u^2 u_{xx} + \sigma\,u_{yy}.
\label{eq:Delta-expanded}
\end{equation}

The nonlinear terms of \eqref{eq:Delta-expanded} form a second
$x$-derivative. Setting
\begin{equation}
\mathcal{H}[u]
= e\,u\,u_{xx} + \tfrac{1}{2}(f-e)\,u_x^2 + \tfrac{1}{3}g\,u^3,
\label{eq:H-def}
\end{equation}
a direct computation gives
\begin{equation}
D_x^2\,\mathcal{H}[u]
= (e+f)\,u_x u_{xxx} + e\,u\,u_{xxxx} + f\,u_{xx}^2
+ 2g\,u\,u_x^2 + g\,u^2 u_{xx},
\label{eq:H-second-derivative}
\end{equation}
so that
\begin{equation}
\Delta = u_{tx} + c\,u_{xxxx} + u_{xxxxxx}
+ D_x^2\,\mathcal{H}[u] + \sigma\,u_{yy}.
\label{eq:Delta-divergence}
\end{equation}
Equivalently, $\Delta = D_x\mathcal{F}[u] + \sigma\,u_{yy}$, where
\begin{equation}
\mathcal{F}[u]
= u_t + c\,u_{xxx} + u_{xxxxx}
+ e\,u\,u_{xxx} + f\,u_x u_{xx} + g\,u^2 u_x .
\label{eq:F-def}
\end{equation}
The form \eqref{eq:Delta-divergence} is used throughout the
multiplier calculations.
\subsection{Equivalence transformations}

The normalized family \eqref{eq:normalized-equation} is preserved by
the point transformations
\begin{equation*}
x=\lambda\widetilde{x}+x_0,
\qquad
y=\varepsilon\lambda^3\widetilde{y}+y_0,
\qquad
t=\lambda^5\widetilde{t}+t_0,
\qquad
u=\rho\widetilde{u},
\end{equation*}
where
\[
\lambda\rho\neq0,
\qquad
\varepsilon=\pm1.
\]
The transformed coefficients are
\begin{equation*}
\widetilde{c}=\lambda^2c,
\qquad
\widetilde{e}=\rho\lambda^2e,
\qquad
\widetilde{f}=\rho\lambda^2f,
\qquad
\widetilde{g}=\rho^2\lambda^4g,
\end{equation*}
while $\sigma$ is unchanged.

Indeed,
\[
D_x=\lambda^{-1}D_{\widetilde{x}},
\qquad
D_y=\varepsilon\lambda^{-3}D_{\widetilde{y}},
\qquad
D_t=\lambda^{-5}D_{\widetilde{t}},
\]
and the terms $u_{tx}$, $u_{xxxxxx}$, and $u_{yy}$ acquire the common
factor $\rho\lambda^{-6}$. Division by this factor gives the stated
coefficient transformations.

When $e\neq0$, the ratios
\begin{equation*}
\frac{f}{e},
\qquad
\frac{g}{e^2}
\end{equation*}
are invariants of these scalings. Hence the two exceptional nonlinear
cases obtained below are characterized invariantly by
\begin{equation*}
\left(
\frac{f}{e},\frac{g}{e^2}
\right)
=
\left(\frac12,0\right)
\qquad\text{and}\qquad
\left(
\frac{f}{e},\frac{g}{e^2}
\right)
=
(3,0).
\end{equation*}

On the subfamily $g=0$, the translation
\begin{equation*}
u=\widetilde{u}+\kappa
\end{equation*}
also preserves the family. It leaves $e$ and $f$ unchanged and
transforms $c$ according to
\begin{equation*}
\widetilde{c}=c+e\kappa.
\end{equation*}
Consequently, when $g=0$ and $e\neq0$, the coefficient $c$ can be
normalized to zero by taking $\kappa=-c/e$. This translation does not
preserve the stated family when $g\neq0$, since additional nonlinear
terms are then generated.
\section{Conservation-law multipliers}
\label{sec:multiplier-framework}
A local conservation law of \eqref{eq:normalized-equation} is a
continuity equation
\begin{equation}
D_tT+D_xX+D_yY=0
\label{eq:conservation-law}
\end{equation}
holding for all solutions, where the conserved density $T$ and the
spatial fluxes $X,Y$ are local differential functions of
$x,y,t,u$, and finitely many derivatives of $u$. Two conservation
laws are locally equivalent if the difference of their conserved
currents is locally trivial, namely, if it consists of a current that
vanishes on the solution space together with a null divergence.

For a regular normal partial differential equation, every local
conservation law has an equivalent characteristic form
\begin{equation}
D_tT+D_xX+D_yY=Q\,\Delta,
\label{eq:characteristic-form}
\end{equation}
holding identically on the full jet space
\cite{AncoBluman2002a,AncoBluman2002b,
BlumanCheviakovAnco2010,Olver1993}. The differential function $Q$ is
called a conservation-law multiplier. It is local and nonsingular on
the solution space.

A differential function $Q$ is a multiplier if and only if
\begin{equation}
E_u\!\left(Q\,\Delta\right)=0
\label{eq:multiplier-equation}
\end{equation}
holds identically, where
\begin{equation*}
E_u
=
\sum_{i,j,k\geq0}
(-D_x)^i(-D_y)^j(-D_t)^k
\frac{\partial}{\partial u_{x^iy^jt^k}}
\end{equation*}
is the Euler operator with respect to $u$. This condition follows
from the fact that the kernel of the Euler operator consists of total
divergences.

\begin{remark}
\label{rem:adjoint-symmetries}
Conservation-law multipliers are closely related to adjoint
symmetries. Every multiplier satisfies the adjoint-symmetry
determining equation on the solution space. Conversely, an adjoint
symmetry is a conservation-law multiplier only when it also satisfies
the associated Helmholtz-type conditions. The condition
\[
E_u\!\left(Q\Delta\right)=0,
\]
imposed identically in jet space, incorporates both the
adjoint-symmetry determining equation and these additional
conditions. Consequently, solving only the adjoint-symmetry equation
may produce differential functions that do not correspond to local
conservation laws. This distinction clarifies the relation between
the direct multiplier method and approaches based on adjoint
symmetries or nonlinear self-adjointness
\cite{Ibragimov2011,IbragimovAvdonina2013,Ma2018}.
\end{remark}

Equation~\eqref{eq:normalized-equation} is a regular normal scalar
partial differential equation. Indeed, it can be solved explicitly
for the leading derivative $u_{tx}$, and the resulting right-hand
side contains neither $u_{tx}$ nor any of its differential
consequences. Hence the standard multiplier--conservation-law
correspondence applies
\cite{AncoBluman2002a,AncoBluman2002b}. Modulo local equivalence,
nontrivial local conservation laws are in one-to-one correspondence
with multipliers restricted to the solution space. In particular,
\begin{equation*}
Q\big|_{\mathcal{E}}\neq0
\quad\Longrightarrow\quad
[T,X,Y]\neq0,
\end{equation*}
where $\mathcal{E}$ denotes the solution space and $[T,X,Y]$ denotes
the local equivalence class of the conserved current. We therefore
classify the conservation laws by classifying their multipliers.

The determining equation \eqref{eq:multiplier-equation} is imposed as
an identity on the full jet space. In particular, $u_{tx}$ and its
differential consequences are retained during the Euler-operator
calculation. The normal form is used only to specify the admissible
normal jet variables in the multiplier and to establish the
multiplier--conservation-law correspondence.

The differential order of a multiplier is the highest total order of
a derivative of $u$ on which it depends. The classification below is
restricted to nonlinear members of
\eqref{eq:normalized-equation},
\begin{equation*}
(e,f,g)\neq(0,0,0),
\end{equation*}
and to multipliers of differential order at most two that are local
in the original dependent variable $u$. Thus no potential variable
or inverse total derivative is introduced. Conservation laws that are
local only after introducing a potential, point symmetries, and
variational formulations are not considered. The linear member $e=f=g=0$ lies outside the present classification.

\section{Reduction to first order}
\label{sec:order-reduction}

We first show that a multiplier of low differential order cannot
depend on the highest jets, for any values of the coefficients.
Since equation~\eqref{eq:normalized-equation} is solved for the
leading derivative $u_{tx}$, this derivative and its differential
consequences are principal derivatives and are excluded from the
multiplier ansatz. The ansatz is therefore expressed in terms of the
parametric derivatives. At second order, the relevant parametric
derivatives are
\[
u_{xx},
\qquad
u_{xy},
\qquad
u_{yy},
\qquad
u_{yt},
\qquad
u_{tt}.
\]

\begin{theorem}
\label{thm:order-reduction}
Let $Q$ be a multiplier of \eqref{eq:normalized-equation} of the form
\begin{equation}
Q = Q\bigl(x,y,t,u,u_x,u_y,u_t,
u_{xx},u_{xy},u_{yy},u_{yt},u_{tt}\bigr).
\label{eq:second-order-ansatz}
\end{equation}
Then $Q$ is independent of every second-order jet:
\begin{equation}
Q_{u_{xx}} = Q_{u_{xy}} = Q_{u_{yy}}
= Q_{u_{yt}} = Q_{u_{tt}} = 0,
\label{eq:no-second-order}
\end{equation}
so that $Q = Q(x,y,t,u,u_x,u_y,u_t)$. This holds for all
$c,e,f,g$ and $\sigma=\pm1$.
\end{theorem}

\begin{proof}
The determining equation $E_u(Q\Delta)=0$ is an identity in the jet
variables. Let $u_J$ be any one of the five second-order jets in
\eqref{eq:second-order-ansatz}, so that $|J|=2$, and consider in
$E_u(Q\Delta)$ the coefficient of the single jet $u_{J+(6,0,0)}$,
which has total order eight.

Two contributions to this coefficient arise. The Euler operator
contains the term
\[
(-D)_J\left[
\frac{\partial(Q\Delta)}{\partial u_J}
\right].
\]
Since $u_{xxxxxx}$ occurs in $\Delta$ linearly with coefficient one,
$\partial(Q\Delta)/\partial u_J$ contains $Q_{u_J}\,u_{xxxxxx}$, and
applying $(-D)_J$ produces $Q_{u_J}\,u_{J+(6,0,0)}$ with sign
$(-1)^{|J|}=+1$. The Euler operator also contains
\[
(-D_x)^6\left[
\frac{\partial(Q\Delta)}{\partial u_{xxxxxx}}
\right]
=
(-D_x)^6[Q].
\] The term in which all six derivatives act on $Q$
through its dependence on $u_J$ produces $Q_{u_J}\,u_{J+(6,0,0)}$ with
sign $(-1)^6=+1$.

No further contribution to the single jet $u_{J+(6,0,0)}$ occurs.
Since $Q$ has differential order at most two, it can supply a jet of
order eight only through a single differentiation of its dependence on
a second-order jet, as in the two contributions above; any other
Euler term either acts on a jet of $\Delta$ of $x$-order at most four,
yielding single jets of total order at most six, or arises from a
quadratic term of $\Delta$, yielding products of jets rather than the
single jet $u_{J+(6,0,0)}$. Hence the coefficient of $u_{J+(6,0,0)}$
equals $2Q_{u_J}$, and the determining equation forces $Q_{u_J}=0$.
Applying this to each of the five second-order jets gives
\eqref{eq:no-second-order}. No division by the coefficients is used,
so the conclusion is uniform in $c,e,f,g,\sigma$.
\end{proof}
\section{Classification of the multipliers}
\label{sec:classification}

By Theorem~\ref{thm:order-reduction} it suffices to consider
\begin{equation}
Q = Q(x,y,t,u,u_x,u_y,u_t).
\label{eq:first-order-ansatz}
\end{equation}
Throughout this section the equation is nonlinear, namely
\begin{equation*}
(e,f,g)\neq(0,0,0).
\end{equation*}

\subsection{Reduction of the ansatz}

\begin{lemma}
\label{lem:first-order-reduction}
Every multiplier of the form \eqref{eq:first-order-ansatz} of a
nonlinear member of the family \eqref{eq:normalized-equation} has the
form
\begin{equation}
Q = A(t)\,u_x + K(x,y,t).
\label{eq:reduced-form}
\end{equation}
\end{lemma}

\begin{proof}
Splitting $E_u(Q\Delta)=0$ with respect to the highest jet monomials
shows first that $Q$ is affine in $u_x,u_y,u_t$, so that
\begin{equation}
Q = A\,u_x + B\,u_y + C\,u_t + R,
\label{eq:affine-form}
\end{equation}
with $A,B,C,R$ functions of $x,y,t,u$. Further coefficients give
$A_u=B_u=C_u=0$ and $B_x=C_x=C_y=0$, together with
\begin{equation}
2R_u = B_y,
\qquad
B_y = 3A_x,
\qquad
C_t = 5A_x .
\label{eq:parameter-free-relations}
\end{equation}
The coefficients of the remaining jet monomials yield
\begin{equation}
(e+f)B = (e-f)B = 0,
\qquad
(e+f)C = (e-f)C = 0,
\qquad
gB = gC = 0 .
\label{eq:BC-conditions}
\end{equation}
If $(e,f)\neq(0,0)$ then $e+f$ and $e-f$ cannot both vanish, so
$B=C=0$; if $e=f=0$ then $g\neq0$ and again $B=C=0$. In either case
\eqref{eq:parameter-free-relations} gives $A_x=0$ and $R_u=0$, and a
further coefficient gives $A_y=0$. Hence $A=A(t)$ and $R=K(x,y,t)$,
which is \eqref{eq:reduced-form}. The determining equation was set up, expanded, and split
symbolically using Maple.
\end{proof}

\subsection{The determining system}

Substituting \eqref{eq:reduced-form} into $E_u(Q\Delta)$ gives
\begin{equation}
\begin{aligned}
E_u(Q\Delta)
={}& \mathcal{L}[K]
+ e\,u\,K_{xxxx}
+ (3e-f)\,u_x K_{xxx}
+ \bigl(A' + (3e-f)K_{xx}\bigr) u_{xx}
+ g\,u^2 K_{xx}
\\
&+ (e-2f)\,A\,\bigl(u_x u_{xxxx} + 2u_{xx}u_{xxx}\bigr)
- 2g\,A\,\bigl(u_x^3 + 3u\,u_x u_{xx}\bigr),
\end{aligned}
\label{eq:reduced-euler}
\end{equation}
where
\begin{equation}
\mathcal{L}[K]
= K_{xt} + c\,K_{xxxx} + K_{xxxxxx} + \sigma K_{yy}.
\label{eq:L-operator}
\end{equation}
Splitting \eqref{eq:reduced-euler} with respect to the independent
jet monomials and the powers of $u$ gives the determining system
\begin{equation}
(e-2f)A = 0,
\qquad
g\,A = 0,
\qquad
A' + (3e-f)K_{xx} = 0,
\label{eq:det-A}
\end{equation}
\begin{equation}
e\,K_{xxxx} = 0,
\qquad
(3e-f)K_{xxx} = 0,
\qquad
g\,K_{xx} = 0,
\qquad
\mathcal{L}[K] = 0 .
\label{eq:det-K}
\end{equation}

By \eqref{eq:det-A}, a multiplier depends on $u_x$ only if $g=0$ and
$e=2f$. By \eqref{eq:det-K}, $K$ fails to be affine in $x$ only if
$g=0$ and $f=3e$. These two conditions are incompatible for a
nonlinear equation: $e=2f$ and $f=3e$ give $e=f=0$, which together
with $g=0$ is the linear case. The classification therefore splits
into three mutually exclusive cases.

\subsection{The generic case}

\begin{theorem}
\label{thm:generic}
Suppose that either $g\neq0$, or $g=0$ with $e\neq2f$ and $f\neq3e$.
Then every multiplier of order at most two is
\begin{equation}
Q = x\bigl(a_0(t) + a_1(t)\,y\bigr)
+ b_0(t) + b_1(t)\,y
- \frac{a_0'(t)\,y^2}{2\sigma}
- \frac{a_1'(t)\,y^3}{6\sigma},
\label{eq:generic-multiplier}
\end{equation}
where $a_0,a_1,b_0,b_1$ are arbitrary functions of $t$.
\end{theorem}

\begin{proof}
If $g\neq0$, then $gA=0$ and $gK_{xx}=0$ give $A=0$ and $K_{xx}=0$.
If $g=0$ and $e\neq2f$, then $(e-2f)A=0$ gives $A=0$, whence
$(3e-f)K_{xx}=0$ and $f\neq3e$ give $K_{xx}=0$. In both cases
$K_{xxxx}=K_{xxxxxx}=0$, so $\mathcal{L}[K]=0$ reduces to
\begin{equation}
K_{xx}=0,
\qquad
K_{xt} + \sigma K_{yy} = 0 .
\label{eq:generic-system}
\end{equation}
Writing $K = x\,a(y,t) + b(y,t)$ and splitting the second equation in
powers of $x$ gives $a_{yy}=0$ and $a_t + \sigma b_{yy} = 0$.
Integrating in $y$ yields \eqref{eq:generic-multiplier}.
\end{proof}
\subsection{The case \texorpdfstring{$g=0$, $e=2f$}{g=0, e=2f}}
\begin{remark}
\label{rem:e2f-reduction}
This case is distinguished by the one-dimensional reduction. For the
$y$-independent reduction of \eqref{eq:normalized-equation} with
$g=0$, integration by parts gives
\begin{equation}
\frac{d}{dt}\int \tfrac12 u_x^2 \,dx
= \left(f - \tfrac12 e\right)\int u_x u_{xx}^2 \,dx
\label{eq:1d-gradient-energy}
\end{equation}
for solutions with suitable decay or periodicity. The density
$\tfrac12 u_x^2$ is thus conserved in one dimension precisely when
$e=2f$, which is the condition under which $u_x$ occurs as a
multiplier in two dimensions.
\end{remark}
\begin{theorem}
\label{thm:e-2f}
Suppose $g=0$, $e=2f$ and $f\neq0$. Then every multiplier of order at
most two is
\begin{equation}
Q = A(t)\,u_x + K(x,y,t),
\label{eq:e2f-multiplier}
\end{equation}
where
\begin{equation}
\begin{aligned}
K = {}&
-\frac{A'(t)}{10f}\,x^2
+ \left(
\frac{A''(t)\,y^2}{10\sigma f}
+ p_1(t)\,y + p_0(t)
\right) x
\\
&- \frac{A'''(t)\,y^4}{120\,\sigma^2 f}
- \frac{p_1'(t)\,y^3}{6\sigma}
- \frac{p_0'(t)\,y^2}{2\sigma}
+ r_1(t)\,y + r_0(t),
\end{aligned}
\label{eq:e2f-K}
\end{equation}
and $A,p_0,p_1,r_0,r_1$ are arbitrary functions of $t$. In
particular, $Q=u_x$ is a multiplier.
\end{theorem}

\begin{proof}
Here $3e-f=5f\neq0$, so \eqref{eq:det-K} gives $K_{xxx}=0$, hence
$K_{xxxx}=K_{xxxxxx}=0$, and the determining system reduces to
\begin{equation}
K_{xxx}=0,
\qquad
A' + 5f\,K_{xx} = 0,
\qquad
K_{xt} + \sigma K_{yy} = 0 .
\label{eq:e2f-system}
\end{equation}
Writing $K = k_2(y,t)x^2 + k_1(y,t)x + k_0(y,t)$, the second equation
gives $k_2 = -A'/(10f)$. Splitting the third equation in powers of
$x$ then gives $k_{2,yy}=0$, $2k_{2,t} + \sigma k_{1,yy} = 0$ and
$k_{1,t} + \sigma k_{0,yy} = 0$; successive integration in $y$ yields
\eqref{eq:e2f-K}. Taking $A=1$ and
$p_0=p_1=r_0=r_1=0$ gives $Q=u_x$.
\end{proof}
\subsection{The case \texorpdfstring{$g=0$, $f=3e$}{g=0, f=3e}}
\begin{remark}
\label{rem:f3e-divergence}
This case is distinguished by the structure of the nonlinear terms.
When $g=0$ and $f=3e$,
\begin{equation}
\mathcal{H}[u] = e\bigl(u\,u_{xx} + u_x^2\bigr)
= \tfrac12 e\,D_x^2(u^2),
\label{eq:H-collapse}
\end{equation}
so that by \eqref{eq:Delta-divergence} the nonlinear part of $\Delta$
is the fourth $x$-derivative $\tfrac12 e\,D_x^4(u^2)$. The multiplier
pairs with this term through $K_{xxxx}$ rather than through
$K_{xx}$, which is why the multiplier may be cubic in $x$ on this
case and is affine in $x$ otherwise.
\end{remark}
\begin{theorem}
\label{thm:f-3e}
Suppose $g=0$, $f=3e$ and $e\neq0$. Then every multiplier of order at
most two is independent of $u$ and its derivatives, and is given by
\begin{equation}
\begin{aligned}
Q = {}&
x^3\bigl(a_0 + a_1 y\bigr)
+ x^2\left(
b_0 + b_1 y
- \frac{3a_0'\,y^2}{2\sigma}
- \frac{a_1'\,y^3}{2\sigma}
\right)
\\
&+ x\left(
d_0 + d_1 y
- \frac{b_0'\,y^2}{\sigma}
- \frac{b_1'\,y^3}{3\sigma}
+ \frac{a_0''\,y^4}{4\sigma^2}
+ \frac{a_1''\,y^5}{20\sigma^2}
\right)
\\
&+ r_0 + r_1 y
- \frac{d_0'\,y^2}{2\sigma}
- \frac{d_1'\,y^3}{6\sigma}
+ \frac{b_0''\,y^4}{12\sigma^2}
+ \frac{b_1''\,y^5}{60\sigma^2}
- \frac{a_0'''\,y^6}{120\sigma^3}
- \frac{a_1'''\,y^7}{840\sigma^3},
\end{aligned}
\label{eq:f3e-multiplier}
\end{equation}
where $a_0,a_1,b_0,b_1,d_0,d_1,r_0,r_1$ are arbitrary functions
of~$t$.
\end{theorem}

\begin{proof}
Here $e-2f=-5e\neq0$, so $A=0$ by \eqref{eq:det-A}. Since $3e-f=0$,
no condition on $K_{xx}$ or $K_{xxx}$ remains, while $e\neq0$ and
\eqref{eq:det-K} give $K_{xxxx}=0$. Then $K_{xxxxxx}=0$ and
$\mathcal{L}[K]=0$ reduces to
\begin{equation}
K_{xxxx}=0,
\qquad
K_{xt} + \sigma K_{yy} = 0 .
\label{eq:f3e-system}
\end{equation}
Writing $K = q_3x^3 + q_2x^2 + q_1x + q_0$ and splitting the second
equation in powers of $x$ gives $q_{3,yy}=0$,
$3q_{3,t} + \sigma q_{2,yy} = 0$, $2q_{2,t} + \sigma q_{1,yy} = 0$ and
$q_{1,t} + \sigma q_{0,yy} = 0$; successive integration in $y$ yields
\eqref{eq:f3e-multiplier}.
\end{proof}

\subsection{Summary of the classification}
\label{subsec:classification-summary}
The preceding results can be collected into the following complete
classification.

\begin{theorem}[Complete low-order classification]
\label{thm:complete-classification}

For every nonlinear member of the generalized fifth-order KP family
\eqref{eq:normalized-equation}, all local conservation laws whose
multipliers are local in $u$ and have differential order at most two
are classified, up to local equivalence, by the following three
parameter cases.

\begin{enumerate}[label=\textnormal{(\roman*)}]

\item For all nonlinear parameter values outside the two exceptional
cases below, every multiplier has the form
\begin{equation*}
Q=K(x,y,t),
\end{equation*}
where
\begin{equation*}
K_{xx}=0,
\qquad
K_{xt}+\sigma K_{yy}=0.
\end{equation*}
This family depends on four arbitrary functions of $t$.

\item When
\begin{equation*}
g=0,
\qquad
e=2f,
\qquad
f\neq0,
\end{equation*}
every multiplier has the form
\begin{equation*}
Q=A(t)u_x+K(x,y,t),
\end{equation*}
where
\begin{equation*}
K_{xxx}=0,
\qquad
A'(t)+5fK_{xx}=0,
\qquad
K_{xt}+\sigma K_{yy}=0.
\end{equation*}
This family depends on five arbitrary functions of $t$.

\item When
\begin{equation*}
g=0,
\qquad
f=3e,
\qquad
e\neq0,
\end{equation*}
every multiplier has the form
\begin{equation*}
Q=K(x,y,t),
\end{equation*}
where
\begin{equation*}
K_{xxxx}=0,
\qquad
K_{xt}+\sigma K_{yy}=0.
\end{equation*}
This family depends on eight arbitrary functions of $t$.

\end{enumerate}

Every nonzero multiplier in these families determines a nontrivial
local conservation-law class.

\end{theorem}

\begin{proof}

Theorem~\ref{thm:order-reduction} shows that every multiplier
of differential order at most two reduces to first order.
Lemma~\ref{lem:first-order-reduction} then gives
\begin{equation*}
Q=A(t)u_x+K(x,y,t).
\end{equation*}
The determining equations \eqref{eq:det-A}--\eqref{eq:det-K} split
the nonlinear coefficient space into the three cases stated above.

In the generic case, integrating $K_{xx}=0$ and
$K_{xt}+\sigma K_{yy}=0$ yields four arbitrary functions of $t$.
In the first exceptional case, $K_{xxx}=0$ together with
$A'+5fK_{xx}=0$ and $K_{xt}+\sigma K_{yy}=0$ yields five arbitrary
functions of $t$. In the second exceptional case, $K_{xxxx}=0$ and
$K_{xt}+\sigma K_{yy}=0$ yield eight arbitrary functions of $t$.
The multiplier--conservation-law correspondence stated in
Section~\ref{sec:multiplier-framework} gives the final assertion.

\end{proof}
\begin{remark}
The branching is governed entirely by the determining system
\eqref{eq:det-A}--\eqref{eq:det-K}. The two exceptional cases are
mutually exclusive for a nonlinear equation. Indeed, the simultaneous
relations
\[
e=2f,
\qquad
f=3e
\]
imply $e=f=0$, and together with the exceptional-branch condition
$g=0$ this gives the excluded linear equation.

No additional cases arise when $e=0$, $f\neq0$, $g=0$, or when
$f=0$, $e\neq0$, $g=0$; both belong to the generic case. Likewise,
the relations $e+f=0$ and $e-f=0$, which occur in the first-order
coefficient splitting, do not enlarge the multiplier space.
\end{remark}
\section{Conserved currents}
\label{sec:conserved-currents}

For each multiplier family obtained in
Section~\ref{sec:classification}, we now give a representative
conserved current. Each current is presented in characteristic form,
\[
D_tT+D_xX+D_yY=Q\Delta,
\]
holding identically for arbitrary $u$.

For convenience, we first treat the generic case and the exceptional
case whose multipliers remain zeroth order, and then treat the
exceptional case containing the derivative-dependent term
$A(t)u_x$.

For differential functions $P$ and $R$
and an even positive integer $n$, put
\begin{equation}
\mathcal{B}_n(P,R)
= \sum_{j=0}^{n-1} (-1)^j\,D_x^{\,j}(P)\,D_x^{\,n-1-j}(R),
\label{eq:B-def}
\end{equation}
so that
\begin{equation}
D_x\,\mathcal{B}_n(P,R)
=
P\,D_x^{\,n}R-D_x^{\,n}(P)\,R.
\label{eq:B-identity}
\end{equation}

\subsection{An identity for the zeroth-order multipliers}

\begin{proposition}
\label{prop:universal-current}
For every function $K(x,y,t)$,
\begin{equation}
\begin{aligned}
K\,\Delta
= {}& D_t\bigl(K\,u_x\bigr)
+ D_x\Bigl(
-K_t\,u
+ c\,\mathcal{B}_4(K,u)
+ \mathcal{B}_6(K,u)
+ \mathcal{B}_2\bigl(K,\mathcal{H}[u]\bigr)
\Bigr)
\\
&+ D_y\Bigl(\sigma\bigl(K\,u_y - K_y\,u\bigr)\Bigr)
+ u\,\mathcal{L}[K]
+ K_{xx}\,\mathcal{H}[u],
\end{aligned}
\label{eq:universal-current}
\end{equation}
as an identity in the jet variables.
\end{proposition}

\begin{proof}
The mixed term gives
$K\,u_{tx} = D_t(K u_x) - D_x(K_t u) + K_{xt}\,u$, and the
transverse term gives
$\sigma K u_{yy} = D_y\bigl(\sigma(K u_y - K_y u)\bigr)
+ \sigma K_{yy}\,u$. Applying \eqref{eq:B-identity} with $n=4$ and
$n=6$ to the terms $c\,K u_{xxxx}$ and $K u_{xxxxxx}$, and with $n=2$
to $K\,D_x^2\mathcal{H}[u]$, and adding the results using the form
\eqref{eq:Delta-divergence} of $\Delta$, gives
\eqref{eq:universal-current}.
\end{proof}

The last two terms in \eqref{eq:universal-current} are handled
differently in the three nonlinear cases. In the generic case,
$K_{xx}=0$ and $\mathcal{L}[K]=0$, so both terms vanish. In the case
$g=0$, $f=3e$, the relation
\[
\mathcal{H}[u]=\frac{e}{2}D_x^2(u^2)
\]
and the condition $K_{xxxx}=0$ imply that
$K_{xx}\mathcal{H}[u]$ is a total $x$-derivative. In the case
$g=0$, $e=2f$, the term $K_{xx}\mathcal{H}[u]$ yields a total
$x$-derivative together with a residual term that cancels the
corresponding residual arising from the $A(t)u_x$ part of the
multiplier. The resulting conserved currents are obtained below.

\subsection{The generic case}

Let $K$ be the multiplier \eqref{eq:generic-multiplier}. Then
$K_{xx}=0$, so all higher $x$-derivatives of $K$ vanish and
$\mathcal{L}[K] = K_{xt}+\sigma K_{yy} = 0$. Hence
\eqref{eq:universal-current} gives the conserved current
\begin{equation}
\begin{aligned}
T &= K\,u_x,
\\
X &=
-K_t\,u
+c\,\mathcal{B}_4(K,u)
+\mathcal{B}_6(K,u)
+\mathcal{B}_2\bigl(K,\mathcal{H}[u]\bigr),
\\
Y &= \sigma\bigl(K\,u_y-K_y\,u\bigr),
\end{aligned}
\label{eq:generic-current}
\end{equation}
where
\[
\mathcal{B}_2\bigl(K,\mathcal{H}[u]\bigr)
=
K\,D_x\mathcal{H}[u]-K_x\,\mathcal{H}[u].
\]

Since $K_{xx}=0$ and
$K_{xt}+\sigma K_{yy}=0$, the current satisfies the characteristic
identity
\[
D_tT+D_xX+D_yY=K\Delta
\]
identically for arbitrary $u$.

Modulo a total $x$-derivative the density is equivalent to
\begin{equation}
T \sim -K_x\,u ,
\label{eq:generic-density-reduced}
\end{equation}
since $K u_x = D_x(K u) - K_x u$. For the multiplier
\eqref{eq:generic-multiplier} this gives
$T \sim -\bigl(a_0(t)+a_1(t)y\bigr)u$.

\begin{remark}
\label{rem:zero-density}
The multipliers with $a_0=a_1=0$ have $K_x=0$, so their densities are
equivalent to zero and the resulting conservation laws take the
purely spatial form $D_x X + D_y Y = 0$ on solutions. By the
correspondence recalled in Section~\ref{sec:multiplier-framework},
these conservation laws are nevertheless nontrivial, since their
multipliers do not vanish on solutions.
\end{remark}

\subsection{The case \texorpdfstring{$g=0$, $f=3e$}{g=0, f=3e}}

Let $K$ be the multiplier \eqref{eq:f3e-multiplier}. Here
$K_{xxxx}=0$ and $\mathcal{L}[K]=K_{xt}+\sigma K_{yy}=0$, while by
Remark~\ref{rem:f3e-divergence} $\mathcal{H}[u]=\tfrac12 e\,D_x^2(u^2)$.
Two applications of \eqref{eq:B-identity} give
\begin{equation}
K_{xx}\,\mathcal{H}[u]
= \tfrac12 e\,K_{xx}\,D_x^2(u^2)
= \tfrac12 e\,D_x\Bigl(
K_{xx}\,D_x(u^2) - K_{xxx}\,u^2
\Bigr)
+ \tfrac12 e\,K_{xxxx}\,u^2 ,
\label{eq:f3e-residual}
\end{equation}
and the last term vanishes. Hence \eqref{eq:universal-current} gives
the conserved current
\begin{equation}
T = K\,u_x,
\qquad
X = -K_t\,u
+ c\,\mathcal{B}_4(K,u)
+ \mathcal{B}_6(K,u)
+ \tfrac12 e\,\mathcal{B}_4\bigl(K,u^2\bigr),
\qquad
Y = \sigma\bigl(K\,u_y - K_y\,u\bigr).
\label{eq:f3e-current}
\end{equation}
The relations $K_{xxxx}=0$ and
$K_{xt}+\sigma K_{yy}=0$, together with
\eqref{eq:f3e-residual}, therefore give the characteristic identity
\[
D_tT+D_xX+D_yY=K\Delta,
\]
holding identically for arbitrary $u$.
Since $K$ is cubic in $x$, the reduced density $-K_x u$ is quadratic
in $x$, so this case yields conservation laws with densities
involving the second $x$-moment of $u$.

\subsection{The case \texorpdfstring{$g=0$, $e=2f$}{g=0, e=2f}}

Let $Q = A(t)u_x + K$ be the multiplier of
Theorem~\ref{thm:e-2f}. Set
\begin{equation*}
\begin{aligned}
T_K &= K\,u_x,
\\
X_K &=
-K_t\,u
+c\,\mathcal{B}_4(K,u)
+\mathcal{B}_6(K,u)
+\mathcal{B}_2\bigl(K,\mathcal{H}[u]\bigr),
\\
Y_K &= \sigma\bigl(K\,u_y-K_y\,u\bigr).
\end{aligned}
\end{equation*}
On this branch,
\begin{equation*}
\mathcal{H}[u]
=
f\left(2u\,u_{xx}-\tfrac12u_x^2\right).
\end{equation*}
Moreover, $K_{xxx}=0$, $K_{xt}+\sigma K_{yy}=0$, and
$A'+5fK_{xx}=0$. Hence
\begin{align*}
K_{xx}\mathcal{H}[u]
&=
-\frac{A'}{5f}
f\left(2u\,u_{xx}-\tfrac12u_x^2\right)
\\
&=
-\frac25A'\,u\,u_{xx}
+\frac1{10}A'\,u_x^2
\\
&=
D_x\left(-\frac25A'\,u\,u_x\right)
+\frac12A'\,u_x^2.
\end{align*}
Consequently, the universal identity
\eqref{eq:universal-current} gives
\begin{equation}
K\,\Delta
=
D_tT_K
+
D_x\left(
X_K-\frac25A'(t)\,u\,u_x
\right)
+
D_yY_K
+
\frac12A'(t)\,u_x^2.
\label{eq:e2f-K-part}
\end{equation} The part
$A(t)u_x$ satisfies
\begin{equation}
A(t)\,u_x\,\Delta
=
D_t\Bigl(\tfrac12A\,u_x^2\Bigr)
+
D_x\bigl(A\,\mathcal{P}[u]\bigr)
+
D_y\bigl(\sigma A\,u_xu_y\bigr)
-
\tfrac12A'(t)\,u_x^2,
\label{eq:e2f-A-part}
\end{equation}
where
\begin{equation}
\begin{aligned}
\mathcal{P}[u]
= {}& c\Bigl(u_x u_{xxx} - \tfrac12 u_{xx}^2\Bigr)
+ u_x u_{xxxxx} - u_{xx}u_{xxxx} + \tfrac12 u_{xxx}^2
\\
&+ e\,u\,u_x u_{xxx} + f\,u_x^2 u_{xx}
- \tfrac12 e\,u\,u_{xx}^2
- \tfrac12 \sigma\,u_y^2 .
\end{aligned}
\label{eq:P-def}
\end{equation}
The residual term
$\tfrac12A'(t)u_x^2$ in \eqref{eq:e2f-K-part}
cancels the residual term
$-\tfrac12A'(t)u_x^2$ in \eqref{eq:e2f-A-part}.
Therefore, adding the two identities gives the conserved current
\begin{equation}
\begin{aligned}
T
&=
\tfrac12A(t)\,u_x^2+K\,u_x,
\\
X
&=
A\,\mathcal{P}[u]
+X_K
-\frac25A'(t)\,u\,u_x,
\\
Y
&=
\sigma\bigl(
A\,u_xu_y
+K\,u_y
-K_y\,u
\bigr).
\end{aligned}
\label{eq:e2f-current}
\end{equation}
In particular $A=1$, $K=0$ gives the conserved density
$\tfrac12 u_x^2$, with flux $\mathcal{P}[u]$ and
$\sigma u_x u_y$.
\subsection{Equivalent conserved densities}

The conserved densities displayed above can be simplified up to local
equivalence. Since
\begin{equation*}
K u_x=D_x(Ku)-K_xu,
\end{equation*}
the generic current \eqref{eq:generic-current} and the current
\eqref{eq:f3e-current} have the equivalent conserved density
\begin{equation}
\widehat{T}_K=-K_xu.
\label{eq:equivalent-density-K}
\end{equation}
Similarly, the current \eqref{eq:e2f-current} has the equivalent
conserved density
\begin{equation}
\widehat{T}_{A,K}
=
\frac12A(t)u_x^2-K_xu.
\label{eq:equivalent-density-AK}
\end{equation}

\subsection{Interpretation of the conserved quantities}
\label{subsec:interpretation}

The classification gives a direct interpretation of the conserved
densities admitted by the different parameter branches. In the
generic case, the equivalent density is
\[
\widehat{T}_K=-K_xu,
\]
where
\[
K_x=a_0(t)+a_1(t)y.
\]
Consequently, the only nonzero reduced densities in the generic case
are linear combinations of $u$ and $yu$. In particular, the choice
$K=x$ gives
\[
\widehat{T}_K=-u,
\]
which represents conservation of the total mass, up to an
inessential overall sign. Similarly, the choice $K=xy$ gives
\[
\widehat{T}_K=-yu,
\]
which represents the first transverse moment of the mass
distribution. The choices $K=1$ and $K=y$ satisfy $K_x=0$ and hence
have zero equivalent density. As discussed in
Remark~\ref{rem:zero-density}, these choices yield nontrivial purely
spatial conservation laws, but they do not produce nonzero conserved
integrals through the density.

The branch
\[
g=0,
\qquad
e=2f,
\qquad
f\neq0,
\]
is distinguished by the appearance of a genuinely nonlinear
conserved density. Taking $A=1$ and $K=0$ gives
\[
\widehat{T}_{A,K}=\frac12u_x^2.
\]
The corresponding integral
\[
\mathcal{E}_x[u]
=
\frac12\iint_{\Omega}u_x^2\,dx\,dy
\]
can be interpreted as a longitudinal gradient energy. This density
does not occur in the generic case or in the branch $f=3e$. When
$A'(t)\neq0$, the determining equation
\[
A'(t)+5fK_{xx}=0
\]
forces $K$ to contain a quadratic term in $x$. Consequently,
$-K_xu$ contains a term proportional to $xu$. Thus this branch also
admits time-dependent identities coupling the first longitudinal
moment to the longitudinal gradient energy and to transverse
moment terms.

The branch
\[
g=0,
\qquad
f=3e,
\qquad
e\neq0,
\]
has no derivative-dependent multiplier, but $K$ may be cubic in
$x$. Hence $K_x$ may be quadratic in $x$, and the equivalent density
$-K_xu$ can contain terms proportional to
\[
u,
\qquad
xu,
\qquad
x^2u,
\]
together with their transverse-weighted counterparts. In particular,
this is the only branch in the classified low-order family that
admits a second longitudinal moment density proportional to $x^2u$.
The enlarged eight-function multiplier family can therefore be
viewed as an enlargement of the hierarchy of polynomially weighted
mass-moment conservation laws.

Another structural consequence of the classification is the absence
of the multiplier $Q=u$. Indeed, every multiplier of differential
order at most two admitted by a nonlinear member of the family has
the form
\[
Q=A(t)u_x+K(x,y,t).
\]
Thus $Q=u$ does not occur within the low-order local multiplier class
considered here. Moreover, the equivalent densities obtained in the
generic and $f=3e$ cases have the form
\[
\widehat{T}_K=-K_xu,
\]
while those obtained in the $e=2f$ case have the form
\[
\widehat{T}_{A,K}
=
\frac12A(t)u_x^2-K_xu.
\]
None of these densities contains a term proportional to $u^2$.
Consequently, no $L^2$-type conserved density occurs among the local
conservation laws represented by multipliers of differential order
at most two in the present classification. This conclusion does not
exclude conservation laws arising from higher-order local
multipliers or from potential or nonlocal formulations.

The moment interpretations require the corresponding weighted
integrals to exist. Thus the mass integral requires integrability of
$u$, while the transverse and longitudinal moment integrals require
additional decay of $yu$, $xu$, or $x^2u$, as appropriate. These
global requirements are separate from the local multiplier
classification.
\subsection{Conserved integrals and boundary conditions}
\label{subsec:conserved-integrals}
Let $\Omega\subset\mathbb{R}^2$ be a fixed spatial domain with
piecewise smooth boundary $\partial\Omega$, and let
$\boldsymbol{n}=(n_x,n_y)$ denote the outward unit normal. For any
conserved current satisfying
\begin{equation*}
D_tT+D_xX+D_yY=0
\end{equation*}
on solutions, integration over $\Omega$ gives the balance equation
\begin{equation}
\frac{d}{dt}
\iint_{\Omega}T\,dx\,dy
=
-\int_{\partial\Omega}
\bigl(Xn_x+Yn_y\bigr)\,ds.
\label{eq:global-balance}
\end{equation}
Consequently, whenever the total flux through the boundary vanishes,
\begin{equation*}
\int_{\partial\Omega}
\bigl(Xn_x+Yn_y\bigr)\,ds=0,
\end{equation*}
the quantity
\begin{equation}
\mathcal{C}[u]
=
\iint_{\Omega}T\,dx\,dy
\label{eq:conserved-integral}
\end{equation}
is independent of time.

The zero-flux condition holds, for example, on a bounded domain under
boundary conditions that make the normal flux vanish. It also holds
for periodic solutions when the conserved current itself is
periodic. The multipliers
$K(x,y,t)$ can depend polynomially on $x$ and $y$, and hence the
associated currents need not be periodic even when $u$ is periodic.

On $\mathbb{R}^2$, the zero-flux condition holds when the solution
and its derivatives have sufficient spatial decay for the boundary
flux at infinity to vanish. For the polynomially weighted densities,
the corresponding weighted integrals must also be finite. Thus the
mass density requires integrability of $u$, the first-moment
densities require integrability of $xu$ or $yu$, and the
second-moment densities require integrability of $x^2u$, together
with sufficient decay of the derivatives appearing in the fluxes.

Under these assumptions, the generic and $f=3e$ cases yield conserved
integrals with density \eqref{eq:equivalent-density-K}, while the
$e=2f$ case yields conserved integrals with density
\eqref{eq:equivalent-density-AK}. The local multiplier
classification itself does not depend on these global assumptions.
\section{Conclusion}
\label{sec:conclusion}

We have obtained a complete classification of the low-order local
conservation laws for the nonlinear generalized fifth-order KP family
\eqref{eq:normalized-equation}. The classification concerns
conservation laws that are local in the original dependent variable
$u$ and whose multipliers have total differential order at most two.

The leading-derivative terms in the multiplier determining equation
show, independently of the parameters, that every multiplier of
differential order at most two reduces to first order. The
unrestricted first-order determining system then reduces every
multiplier for a nonlinear member of the family to the form
\[
Q=A(t)u_x+K(x,y,t).
\]
Solving the resulting determining equations yields one generic case
and two exceptional parameter cases.

For generic nonlinear parameter values, the multiplier is zeroth
order and is determined by
\[
K_{xx}=0,
\qquad
K_{xt}+\sigma K_{yy}=0.
\]
The resulting family depends on four arbitrary functions of time.

The first exceptional case is
\[
g=0,
\qquad
e=2f,
\qquad
f\neq0.
\]
In this case, the additional first-order term $A(t)u_x$ is admitted,
and the multiplier family depends on five arbitrary functions of
time. This exceptional case is characterized by the fact that
$\tfrac12u_x^2$ is a conserved density for the corresponding
one-dimensional reduction.

The second exceptional case is
\[
g=0,
\qquad
f=3e,
\qquad
e\neq0.
\]
Here the multiplier remains zeroth order, but $K$ can be cubic in
$x$, yielding a family depending on eight arbitrary functions of
time. The enlargement of the multiplier family follows from the
identity
\[
\mathcal{H}[u]
=
\frac{e}{2}D_x^2(u^2),
\]
which makes the nonlinear part of the equation a fourth total
$x$-derivative.

The two exceptional cases are distinguished by the invariants of the
equivalence transformations,
\[
\left(
\frac{f}{e},
\frac{g}{e^2}
\right)
=
\left(\frac12,0\right)
\qquad\text{and}\qquad
\left(
\frac{f}{e},
\frac{g}{e^2}
\right)
=
(3,0),
\]
respectively. Consequently, they cannot be mapped into one another by
an equivalence transformation of the normalized family.

Representative conserved densities and spatial fluxes have been
derived for all three cases. Up to local equivalence, the generic and
$f=3e$ cases have conserved density
\[
\widehat{T}_K=-K_xu,
\]
while the $e=2f$ case has conserved density
\[
\widehat{T}_{A,K}
=
\frac12A(t)u_x^2-K_xu.
\]
Under boundary conditions that make the total normal flux vanish,
including periodic settings when the conserved current itself is
periodic, or under sufficient weighted spatial decay, these local
conservation laws yield the conserved integrals described in
Section~\ref{sec:conserved-currents}. The
multiplier classification itself is local and does not depend on any
such global assumptions.

The completeness result is restricted to conservation laws whose
multipliers are local in the original dependent variable $u$ and have
differential order at most two. It does not exclude additional
conservation laws arising from higher-order local multipliers,
potential formulations, or nonlocal variables. The classification of
such conservation laws remains a possible direction for future work.

\appendix

\section{The first-order determining system}
\label{app:determining-system}

The determining equation $E_u(Q\Delta)=0$ for the unrestricted
first-order ansatz \eqref{eq:first-order-ansatz} was set up,
expanded, and split symbolically using Maple. This appendix records the coefficients used in the
proof of Lemma~\ref{lem:first-order-reduction}; the remaining
equations of the split system are consequences of these together with
the reductions they impose. As stated in
Section~\ref{sec:multiplier-framework}, the split is taken with
respect to all jet variables, so that monomials involving $u_{tx}$
and its differential consequences occur as independent coordinates.

Splitting with respect to the highest jet monomials gives the
coefficients in Table~\ref{tab:affine}, from which $Q$ is affine in
$u_x,u_y,u_t$ and may be written
\begin{equation}
Q = A\,u_x + B\,u_y + C\,u_t + R,
\qquad
A,B,C,R \ \text{functions of} \ x,y,t,u .
\label{eq:app-affine-form}
\end{equation}

\begin{table}[h]
\centering
\begin{tabular}{ll}
\hline
Jet monomial & Coefficient \\
\hline
$u_{xxx}\,u_{xxxxx}$ & $15\,Q_{u_xu_x}$ \\
$u_{xy}\,u_{xxxxxx}$ & $4\,Q_{u_xu_y}$ \\
$u_{xt}\,u_{xxxxxx}$ & $4\,Q_{u_xu_t}$ \\
$u_{xy}\,u_{xxxxxy}$ & $6\,Q_{u_yu_y}$ \\
$u_{tt}\,u_{xy}$ & $Q_{u_yu_t}$ \\
$u_{xt}\,u_{xxxxxt}$ & $6\,Q_{u_tu_t}$ \\
\hline
\end{tabular}
\caption{Coefficients giving the affine dependence on
$u_x,u_y,u_t$.}
\label{tab:affine}
\end{table}

Substituting \eqref{eq:app-affine-form} and splitting again gives the
coefficients in Table~\ref{tab:cascade}. Each entry is reduced using
the conclusions of all preceding rows. In particular, the coefficients
of $u_{yy}$ and $u_{xxxxxx}$ use the relation $2R_u=B_y$ obtained
from the coefficient of $u_{xt}$.

\begin{table}[h]
\centering
\begin{tabular}{lll}
\hline
Jet monomial & Coefficient & Conclusion \\
\hline
$u_{xxx}\,u_{xxxx}$ & $35\,A_u$ & $A_u=0$ \\
$u_{xxxy}\,u_{xxx}$ & $20\,B_u$ & $B_u=0$ \\
$u_{xxxt}\,u_{xxx}$ & $20\,C_u$ & $C_u=0$ \\
$u_{xxxxxy}$ & $6\,B_x$ & $B_x=0$ \\
$u_{xxxxxt}$ & $6\,C_x$ & $C_x=0$ \\
$u_{yt}$ & $2\sigma\,C_y$ & $C_y=0$ \\
$u_{xt}$ & $-4g\,C\,u\,u_x - B_y + 2R_u$
 & $g\,C=0$, \ $2R_u=B_y$ \\
$u_{yy}$ & $-\sigma\bigl(C_t-2B_y+A_x\bigr)$
 & $C_t-2B_y+A_x=0$ \\
$u_{xxxxxx}$ & $-C_t+5A_x$ & $C_t=5A_x$ \\
$u_{xy}$ & $-4g\,B\,u\,u_x + B_t + 2\sigma A_y$
 & $g\,B=0$, \ $B_t+2\sigma A_y=0$ \\
$u_{xxx}\,u_{xy}$ & $-(e+f)\,B$ & $(e+f)B=0$ \\
$u_{xx}\,u_{xxy}$ & $3(e-f)\,B$ & $(e-f)B=0$ \\
$u_{xxx}\,u_{xt}$ & $-(e+f)\,C$ & $(e+f)C=0$ \\
$u_{xx}\,u_{xxt}$ & $3(e-f)\,C$ & $(e-f)C=0$ \\
\hline
\end{tabular}
\caption{The first-order cascade. Each coefficient is reduced using
the conclusions of the preceding rows.}
\label{tab:cascade}
\end{table}

The three relations
\begin{equation}
2R_u = B_y,
\qquad
C_t - 2B_y + A_x = 0,
\qquad
C_t = 5A_x
\label{eq:app-three-relations}
\end{equation}
give
\begin{equation}
B_y = 3A_x,
\qquad
C_t = 5A_x,
\qquad
R_u = \tfrac32 A_x ,
\label{eq:app-solved-relations}
\end{equation}
and hold for all values of the coefficients. The last four rows of
Table~\ref{tab:cascade}, together with $gB=gC=0$, give $B=C=0$ for
every nonlinear equation, as in the proof of
Lemma~\ref{lem:first-order-reduction}. Then
\eqref{eq:app-solved-relations} gives $A_x=0$ and $R_u=0$, and
$B_t+2\sigma A_y=0$ gives $A_y=0$. With $A_u=0$ this yields
$A=A(t)$ and $R=K(x,y,t)$, which is \eqref{eq:reduced-form}.

\bibliographystyle{plain}
\bibliography{references}

\end{document}